\documentclass[12pt]{amsart}
\usepackage{amsmath,amssymb,amsthm,amsfonts,stmaryrd,graphicx}
\usepackage{stmaryrd}
\usepackage[english]{babel}
\usepackage{caption}
\newtheorem{theorem}{{\bf Theorem}}[section]
\newtheorem{corollary}{{\bf Corollary }}[section]
\newtheorem{lemma}{{\bf Lemma}}[section]
\newtheorem{proposition}{{\bf Proposition}}[section]

\theoremstyle{definition}
\newtheorem{definition}{{\bf Definition}}[section]
\theoremstyle{remark}

\numberwithin{equation}{section}

\newcommand{\N}{\mathbb{N}}

\newcommand{\be}{\begin{equation}}
\newcommand{\ee}{\end{equation}}
\newcommand{\bea}{\begin{eqnarray}}
\newcommand{\eea}{\end{eqnarray}}
\newcommand{\bd}{\begin{displaymath}}
\newcommand{\ed}{\end{displaymath}}

\newcommand{\op}{ \oplus_q  }
\newcommand{\om}{ \ominus_q  }

\begin{document}

\title[ Summation formula for generalized   discrete $q$-Hermite II polynomials ]{ Summation formula for generalized   discrete $q$-Hermite II polynomials}%
\author{  Sama Arjika  }%
\address{Faculty of Sciences and Technics,   University of Agadez,   Niger}
\email{rjksama2008@gmail.com}%

\begin{abstract}
In this paper, we provide a  family of generalized discrete $q$-Hermite II polynomials  denoted by $\tilde{h}_{n,\alpha}(x,y|q)$.  An explicit relations connecting  them  with   the $q$-Laguerre   and  Stieltjes-Wigert polynomials are obtained. Summation   formula  is derived  by using different analytical means on their generating functions.  
\end{abstract}

\maketitle

\section{Introduction}
\label{intro}
In their paper, \` Alvarez-Nodarse   et al  \cite{Alvarez}, have introduced a $q$-extension of the discrete $q$-Hermite II polynomials  as:
\bea 
\label{mus}
\mathcal{ H}_{2n}^{(\mu)}(x;q):&=&  (-1)^n(q;q)_n\,L_n^{(\mu-1/2)}(x^2;q)\nonumber\\
\\
  \mathcal{H}_{2n+1}^{(\mu)}(x;q):&=& (-1)^n(q;q)_{n}\,x\,L_n^{(\mu+1/2)}(x^2;q)\nonumber
\eea
where $\mu>-1/2$, $L_n^{(\alpha)}(x;q)$ are the $q$-Laguerre polynomials given by
\bea
\label{Lagu}
L_n^{(\alpha)}(x;q):&=&\frac{(q^{\alpha+1};q)_n}{(q;q)_n}\;{}_1\Phi_1\left(\begin{array}{c}
q^{-n}
\\
q^{ \alpha+1} \end{array}\Big| q;-q^{n+\alpha+1}x\right)\nonumber\\
\\
&=&\frac{1}{(q;q)_n}\;{}_2\Phi_1\left(\begin{array}{c}
q^{-n}, -x\\
0 \end{array}\Big| q;q^{n+\alpha+1}x\right)\nonumber  
\eea
with  $(a;q)_0=1$,   $\displaystyle (a;q)_n=\prod_{k=0}^{n-1}(1-aq^k),\, n= 1, 2, \cdots$,  the $q$-shifted factorial, and
\be
 {}_r\Phi_s\left(\begin{array}{c}
q^{-n}, a_2, \cdots, a_r\\
b_1, b_2, \cdots, b_s \end{array}\Big| q;x\right)=\ee
\be
\sum_{k=0}^\infty \left[(-1)^kq^{k(k-1)/2}\right]^{1+s-r}\frac{(q^{-n};q)_k(a_2;q)_k\cdots(a_r;q)_k}{(b_1;q)_k(b_2;q)_k\cdots(b_s;q)_k}\frac{x^k}{(q;q)_k}\nonumber
\ee
  the usual generalized basic or q-hypergeometric function 
%
of degree $n$ in the variable $x$ (see    Slater \cite[Chap. 3]{SLATER},  Srivastava and Karlsson   \cite[p.347, Eq. (272)]{SrivastaKarlsson}  
 for details).  For  $\mu=0$ in (\ref{mus}),  the polynomials 
$ \mathcal{ H}_{n}^{(0)}(x;q)$ correspond to the  discrete $q$-Hermite II polynomials  \cite{AlSalam,ASK}, i.e.,  
$\mathcal{H}_n^{(0)}(x;q^2)=q^{n(n-1)}\tilde{h}_n(x;q).$
They  show that the  polynomials $ \mathcal{H}_{n}^{(\mu)}(x;q)$ satisfy the   orthogonality relation 
\be
\int_{-\infty}^\infty  \mathcal{H}_{n}^{(\mu)}(x;q) \mathcal{H}_{m}^{(\mu)}(x;q) \omega(x) dx=\pi\,q^{-n/2}(q^{1/2};q^{1/2})_n(q^{1/2};q)_{1/2}\;\delta_{nm}
\ee
on the whole real line $\mathbb{R}$ with  respect to the positive weight function $\omega(x)=1/(-x^2;q)_\infty$.  A detailed discussion of the properties of  the polynomials  $ \mathcal{ H}_{n}^{(\mu)}(x;q)$ can be found in \cite{Alvarez}.

Recently,   Saley Jazmat  et al   \cite{ghermite}, introduced  a novel  extension of discrete $q$-Hermite II polynomials by using  new $q$-operators. This extension  is  defined as:
\bea
\label{11a}
\tilde{h}_{2n,\alpha}(x;q)&=& (- 1 )^{n}\, q^{-n(2n-1)} \frac{(q;q)_{2n}}{(q^{2\alpha+2};q^2)_n}\, L_n^{(\alpha)}\left(x^2q^{-2\alpha-1};q^{2}\right)\nonumber\\
\\
\tilde{h}_{2n+1,\alpha}(x;q)&=&(- 1 )^{n}\,q^{-n(2n+1)} \frac{(q;q)_{2n+1}}{(q^{2\alpha+2};q^2)_{n+1}}\,x\,L_{n }^{(\alpha+1)}\left(x^2q^{-2\alpha-1};q^{2}\right).\nonumber
\eea
  For  $\alpha=-1/2$ in (\ref{11a}),  the polynomials $ \tilde{h}_{n,-\frac{1}{2}}(x;q)$ correspond to  the discrete $q$-Hermite II polynomials, i.e.,  $\tilde{h}_{n,-\frac{1}{2}}(x;q)= \tilde{h}_n(x;q).$ The generalized discrete $q$-Hermite II polynomials (\ref{11a})  satisfy the orthogonality relation
\be
\int_{-\infty}^{+\infty} \tilde{h}_{n,\alpha}(x;q)\tilde{h}_{m,\alpha}(x;q)\omega_\alpha(x;q)|x|^{2\alpha+1}d_qx
\ee
\be
=\frac{2q^{-n^2}\,(1-q)(-q,-q,q^2;q^2)_\infty}{(-q^{-2\alpha-1}, -q^{2\alpha+3},q^{2\alpha+2};q^2)_\infty}\frac{(q;q)_{n}^2}{(q;q)_{n,\alpha}}\,\delta_{n,m}\nonumber
\ee
on the whole real line $\mathbb{R}$ with  respect to the positive weight function $\omega_\alpha(x)=1/(-q^{-2\alpha-1}\,x^2;q^2)_\infty$.
 A detailed discussion of the properties of  the polynomials  $  \tilde{h}_{n,\alpha}(x;q)$ can be found in \cite{ghermite}. 

  Srivastava   and  Jain  \cite{SrivastaJain,JainSrivasta},    investigated    multilinear generating functions for   $q$-Hermite, $q$-Laguerre polynomials and other  special functions.   Relevant connections of these   multilinear generating functions with various known results for the classical or $q$-Hermite polynomials are also indicated.   They also proved many combinatorial $q$-series identities by applying the theory of $q$-hypergeometric functions (see  \cite{JainSrivasta}, for more details).

Motivated by  Saley Jazmat's    work \cite{ghermite},   our interest in this paper is to  introduce new family of ``{\it generalized discrete $q$-Hermite II polynomials  (in short gdq-H2P) $ \tilde{h}_{n,\alpha}(x,y|q)$}" which is an extension of  the generalized discrete $q$-Hermite II polynomials  $\tilde{h}_{n,\alpha}(x;q)$ and investigate  summation formulae. 

The paper is organized as follows. In Section 2, we recall   notations   to be used in the sequel. In Section 3, we define   a gdq-H2P  $ \tilde{h}_{n,\alpha}(x,y|q)$ and investigate  several properties.  In Section 4,  we derive  summation and inversion formulae for   gdq-H2P $ \tilde{h}_{n,\alpha}(x,y|q)$. In Section 5,   concluding remarks are given.  

\section{Notations and Preliminaries}
For the convenience of the reader, we provide in this section a summary of the
mathematical notations and definitions used in this paper.  We refer to the general references \cite{G,ASK}   and \cite{ghermite}  for the definitions and notations. Throughout this paper, we assume that $0<q<1,\, \alpha >-1$.

For a complex  number $a$, the $q$-shifted factorials  are defined by:
\be (a;q)_{0}=1;\, (a;q)_{n}=\prod_{k=0}^{n-1}(1-aq^{k}),
n=1, 2, \cdots; \, (a;q)_{\infty}=\prod_{k=0}^{\infty}(1-aq^{k})
\ee
and the  $q$-number    is defined by:
\be 
[n]_{q}={{1-q^n}\over{1-q}}, ~~~  \quad n!_q: =\prod_{k=1}^{n}[k]_{q},\quad 0!_{q}:=1,  ~n\in \N.
\ee
Let $x$ and $y$ be two real or complex numbers, the   Hahn \cite{HahnW}    $q$-addition $\op $ of $x$ and $y$  is given by:
\bea
\label{addition}
 \big(x\op  y\big)^n:&=&(x+y)(x+q y)\ldots (x+q^{n-1}y)\cr
&=&(q;q)_n\sum_{k=0}^n\frac{ q^{({}^k_2)}x^{n-k}y^k}{(q;q)_k(q;q)_{n-k}}
,\quad n\geq 1,\quad \big(x\op  y\big)^0:=1,
\eea
while the $q$-subtraction $\om$ is given by 
\be
\label{additionm}
 \big(x\om  y\big)^n:= \big(x\op (-y)\big)^n.
\ee
The generalized $q$-shifted factorials \cite{ghermite} are defined by
the   recursion relations
\be
[n+1]_{q,\alpha}!= [n+1+\theta_n(2\alpha+1)]_q\,[n]_{q,\alpha}!
\ee
and
\be
(q;q)_{n+1,\alpha}=(1-q)[n+1+\theta_n(2\alpha+1)]_q(q;q)_{n,\alpha},
\ee
where
\bea
\theta_n=\left\{\begin{array}{ll}1& \mbox{ if n even  }  \\
0& \mbox{ if n odd}.\end{array}\right. 
\eea
Remark that, for  $\alpha=-1/2$,  we have
\be
(q;q)_{n,-1/2}=(q;q)_n, \quad [n]_{q,-1/2}!= (1-q)^n(q;q)_n.
\ee
We denote 
\be 
\label{sam1}
(q;q)_{2n,\alpha}=(q^2;q^2)_{n}(q^{2\,\alpha+2};q^2)_n,
\ee
and
\be
\label{sam2}
(q;q)_{2n+1,\alpha}=(q^2;q^2)_{n}(q^{2\,\alpha+2};q^2)_{n+1}.
\ee 
 The two Euler's   $q$-analogues of the  exponential
functions are given by \cite{G}
\be
\label{q-Expo}
  {E}_q(x)
=\displaystyle\sum_{k=0}^{\infty}\displaystyle\frac{q^{({}^k_2)}\,x^{k}}{(q;q)_{k}}=(-x;q)_{\infty}
\ee
and
\be
\label{q-expo}
 e_q(x)
=\displaystyle\sum_{k=0}^{\infty}\displaystyle\frac{x^{k}}{(q;q)_{k}}=\displaystyle\frac{1}{(x;q)_{\infty}},\;\;\;\;|x|<1.
\ee
For $m\geq 1$ and by means of the generalized $q$-shifted factorials,   we define   two generalized $q$-exponential functions  as follows
\be
\label{q-Expo-alpha}
\tilde{E}_{q^m,\alpha}(x):=\displaystyle\sum_{k=0}^{\infty}\displaystyle
\frac{q^{mk(k-1)/2
}\,x^{k}}{(q^m;q^m)_{k,\alpha}},
\ee
and
\be
\label{q-expo-alpha}
\tilde{e}_{q^m,\alpha}(x):=\displaystyle\sum_{k=0}^{\infty}\displaystyle\frac{x^{k}}{(q^m;q^m)_{k,\alpha}},\;\;|x|<1.
\ee
Remark that, for  $m=1$ and $\alpha=-\frac{1}{2}$, we have:
\be
\tilde{E}_{q,\alpha}(x) =  E_{q}(x), \quad  \tilde{e}_{q,\alpha}(x)= e_{q}(x).
\ee
For $m=2,$ the following elementary result is useful in the sequel to establish the summation  formulae for  gdq-H2P:
\be
\label{iinsert} 
 \tilde{e}_{q^2,-\frac{1}{2}} (x)
\tilde{E}_{q^2,-\frac{1}{2}} (y)=\tilde{e}_{q^2,-\frac{1}{2}}(x\oplus_{q^2} y),
\ee
\be
\label{insert} 
\tilde{e}_{q,-\frac{1}{2}}(x)\tilde{E}_{q^{2},-\frac{1}{2}}(-y)=\tilde{e}_{q}(x\ominus _{q,q^{2}}y),\;\;
  \tilde{e}_{q^2,-\frac{1}{2}}(x)\tilde{E}_{q^2,-\frac{1}{2}}(-x)=1,
\ee
where 
\be 
(a\ominus _{q,q^2}b)^{n}:=n!_{q}\sum_{k=0}^{n}\frac{(-1)^{k}q^{k(k-1)}}{(n-k)!_{q}\;k!_{q^2}}a^{n-k}b^{k},\; \;(a\ominus _{q,q^2}b)^{0}:=1.
\ee
\section{Generalized discrete $q$-Hermite II   polynomials}
In this section, we introduce a  sequence of    gdq-H2P $\{\tilde{h}_{n,\alpha}(x,y|q)\}_{n=0}^\infty$. Several properties related to these polynomials  are derived. 
\begin{definition}
For $x,\,y\in\, \mathbb{R},$ the   gdq-H2P $\{\tilde{h}_{n,\alpha}(x,y|q)\}_{n=0}^\infty$ are defined by:
\be 
\label{sama:qdoublyH}
\tilde{h}_{n,\alpha}(x,y|q):=  (q;q)_n\sum_{k=0}^{\lfloor\,n/2\,\rfloor}
 \frac{ (-1)^k q^{-2n k +k(2k+1) }\;  x^{n-2k}\,y^k}{(q;q)_{n-2k,\alpha}\,  (q^2;q^2)_k}
\ee
and
\be
\label{def2}
\tilde{h}_{n,\alpha}(x,0|q):=\frac{(q;q)_n}{(q;q)_{n,\alpha}}x^n.
\ee 
\end{definition}
Remark that,  
\begin{enumerate}
\item  for $\displaystyle   \; y=1$, we get 
\be
\tilde{h}_{n,\alpha}(x,1|q)=\tilde{h}_{n,\alpha}(x;q)
\ee
where $\tilde{h}_{n,\alpha}(x;q)$ is the generalized discrete $q$-Hermite II polynomial  \cite{ghermite};
\item  for  $\displaystyle  \, \alpha=-1/2$ and $ y=1$, we have 
\be
\tilde{h}_{n,-1/2}(x,1|q)=\tilde{h}_{n }(x;q).
\ee
where $\tilde{h}_{n}(x;q)$ is the  discrete $q$-Hermite II polynomial   \cite{AlSalam,ASK}.
\item  Indeed since
\be
\lim_{q\to 1}\frac{(q^a;q)_n}{(1-q)^{n}}=(a)_n 
\ee
one readily verifies that
\be
\lim_{q\to 1}\frac{\tilde{h}_{n,-\frac{1}{2}}(\sqrt{1-q^2}x,1|q)}{(1-q^2)^{n/2}}=\frac{h_n^{\alpha+\frac{1}{2}}(x)}{2^n}
\ee
where $h_n^{\alpha+\frac{1}{2}}(x)$  is the Rosenblum’s generalized Hermite polynomial   \cite{Rosenblum}.
\end{enumerate}
\begin{lemma}
 The following   recursion relation   for   gdq-H2P  $\{\tilde{h}_{n,\alpha}(x,y|q)\}_{n=0}^\infty$  holds true.
\be 
\label{secondses}
 \frac{1-q^{n+1+\theta_n(2\alpha+1)}}{1-q^{n+1}} \tilde{h}_{n+1,\alpha}(x,y|q)
\ee
\be
 =
x\tilde{h}_{n,\alpha}(x,y|q)-y\,q^{-2n+1}(1-q^n) \tilde{h}_{n-1,\alpha}(x,y|q).\nonumber
\ee 
\end{lemma}
\begin{proof}  To prove the    assertion   (\ref{secondses}), we consider  separately   even and  odd cases of the expression
\be
x\tilde{h}_{n,\alpha}(x,y|q)-y\,q^{-2n+1}(1-q^n) \tilde{h}_{n-1,\alpha}(x,y|q) .
\ee 
 For $n$ even, we have:
\be 
\label{dt}
 x\tilde{h}_{2n,\alpha}(x,y|q)= \frac{(q;q)_{2n}}{(q;q)_{2n,\alpha}}x^{2n+1}+ 
(q;q)_{2n}\sum_{k=1}^{ n }
 \frac{ (-1)^k  q^{-2n k +k(2k+1)  }  x^{2n-2k+1}\,y^k}{(q;q)_{2n-2k,\alpha}\,  (q^2;q^2)_k}. \nonumber
\ee
The right-hand side of the last relation can be written as
\be
\label{hs}
\frac{(q;q)_{2n}}{(q;q)_{2n,\alpha}}x^{2n+1}+(q;q)_{2n} \ee
\be
\times\sum_{k=1}^{n}
 \frac{ (-1)^k  q^{ -2k(2n+1)+k(2k+1)  }  x^{2n+1-2k}\,y^k}{(q;q)_{2n+1-2k,\alpha}\,  (q^2;q^2)_k}\left[q^{2k}(1-q^{2n+2+2\alpha-2k})\right].\nonumber
\ee
In the same way, 
\be 
\label{firdt}
 -y\,q^{-4n+1 }\,(1-q^{2n})\,\tilde{h}_{2n-1,\alpha}(x,y|q)  =-y\,q^{-4n+1}\,(q;q)_{2n}\nonumber
\ee
\be
\label{frta}
\times \sum_{k=0}^{n-1}
 \frac{ (-1)^k  q^{ -2 k(2n+1)+k(2k+1) }  x^{2n+1-2(k+1)}\,y^k}{(q;q)_{2n+1-2(k+1),\alpha}\,  (q^2;q^2)_k}.
\ee 
Change k to $k-1$ in (\ref{frta}), one obtains
\be
\label{rhs}
  (q;q)_{2n}\sum_{k=1}^{n}
 \frac{ (-1)^k  q^{ -2 k(2n+1)+k(2k+1) }  x^{2n+1-2k}\,y^k}{(q;q)_{2n+1-2k,\alpha}\,  (q^2;q^2)_{k}}(1-q^{2k}).
\ee
 Then   combining (\ref{hs}) and (\ref{rhs}), we have
\be 
\label{firdt}
 x\tilde{h}_{2n,\alpha}(x,y|q)-y\,q^{-4n+1 }\,(1-q^{2n})\,\tilde{h}_{2n- 1,\alpha} (x,y|q)  = 
\ee
\be
\frac{(q;q)_{2n}}{(q;q)_{2n,\alpha}}x^{2n+1}+(q;q)_{2n} \sum_{k=1}^{n}
 \frac{ (-1)^k  q^{ -2 k(2n+1)+k(2k+1) }  x^{2n+1-2k}\,y^k}{(q;q)_{2n+1-2k,\alpha}\,  (q^2;q^2)_k}\nonumber
\ee
\be
\times \left[q^{2k}(1-q^{2n+2+2\alpha-2k})+(1-q^{2k})\right].\nonumber
\ee
After simplification, it is equal to
\be
\frac{(q;q)_{2n}}{(q;q)_{2n,\alpha}}x^{2n+1}+\nonumber
\ee
\be
(1-q^{2n+2+2\alpha})(q;q)_{2n} \sum_{k=1}^{n}
 \frac{ (-1)^k  q^{ -2 k(2n+1)+k(2k+1) }  x^{2n+1-2k}\,y^k}{(q;q)_{2n+1-2k,\alpha}\,  (q^2;q^2)_k}.\nonumber
\ee
The last expression  can be written as
\be
\label{summe}
 \frac{1-q^{2n+2+2\alpha}}{1-q^{2n+1}}\tilde{h}_{2n+1,\alpha}(x,y|q).
\ee
Summarizing the above calculations in (\ref{firdt})-(\ref{summe}), we get the assertion (\ref{secondses}) for $n$ even. 
 In the  odd case, the proof follows the same steps as the even case.  
\end{proof}
\begin{theorem}
\label{thems}
We have:
\be 
\label{asser1}
\lim_{\alpha\to +\infty}\tilde{h}_{2n,\alpha}(x,y|q)=  q^{-n(2n-1)} (q;q)_{2n}\, (-y)^n\, S_n\left(x^2y^{-1}q^{-1};q^{2}\right)
\ee 
and
\be 
\label{asser2}
\lim_{\alpha\to +\infty}\tilde{h}_{2n+1,\alpha}(x,y|q)= q^{-n(2n+1)} (q;q)_{2n+1}\,x\, (-y)^n\,S_{n }\left(x^2y^{-1}q^{-1};q^{2}\right)
\ee 
where $S_n(x;q)$ are the Stieltjes-Wigert polynomials \cite{ASK}.
\end{theorem}
In order to prove Theorem \ref{thems}, we need the following Lemma. 
\begin{lemma}
\label{thfems}
For $\alpha >-1$, the  sequence of gdq-H2P $\{\tilde{h}_{n,\alpha}(x,y|q)\}_{n=0}^\infty$   can be written in terms of   $q$-Laguerre polynomials $L_n^{(\alpha)}(x;q)$ as
\be 
\label{entefunc}
\tilde{h}_{2n,\alpha}(x,y|q)=  q^{-n(2n-1)} \frac{(q;q)_{2n}}{(q^{2\alpha+2};q^2)_n}\, (-y)^n\, L_n^{(\alpha)}\left(x^2y^{-1}q^{-2\alpha-1};q^{2}\right)
\ee 
and
\be 
\label{entfunc}
\tilde{h}_{2n+1,\alpha}(x,y|q)= q^{-n(2n+1)} \frac{(q;q)_{2n+1}}{(q^{2\alpha+2};q^2)_{n+1}}\,x\, (-y)^n\,L_{n }^{(\alpha+1)}\left(x^2y^{-1}q^{-2\alpha-1};q^{2}\right).
\ee 
\end{lemma}
In order to prove Lemma   \ref{thfems}, we need the following Proposition. 
\begin{proposition}
\label{propo}
For $\alpha >-1$, the sequence of gdq-H2P $\{\tilde{h}_{n,\alpha}(x,y|q)\}_{n=0}^\infty$   can be written in terms of basic hypergeometric functions as
\be 
\label{hyperg1}
\tilde{h}_{n,\alpha}(x,y|q)=   \frac{(q;q)_{n}}{(q;q)_{n,\alpha}}\,x^{n}
\,{}_2\Phi_{1}\left(\begin{array}{c}q^{-n},q^{-n-2\alpha}\\
0\end{array}\Big|\;q^2;\;-\frac{ y\,q^{2\alpha+3}}{x^2}\right). 
\ee
\end{proposition}
\begin{proof}
 In fact, for $n$ even, and by using  
\be 
(q;q)_{2n-2k,\alpha}=(q^2;q^2)_{n -k}(q^{2\alpha+2};q^2)_{n -k},
\ee 
the  gdq-H2P $\tilde{h}_{n,\alpha}(x,y|q)$  defined in (\ref{sama:qdoublyH}) can be rewritten as 
 \be 
\label{samdoublyH}
\tilde{h}_{2n,\alpha}(x,y|q)=
 (q;q)_{2n}\sum_{k=0}^{n}
 \frac{ (-1)^k q^{-4n k +k(2k+1) }  x^{2n-2k}\,y^k}{(q^2;q^2)_{n -k}(q^{2\alpha+2};q^2)_{n -k}\,  (q^2;q^2)_k}. 
\ee 
From the formula  \cite[p.9,  Eq. (0.2.12)]{ASK}  
\be
(a;q)_{n-k}=\frac{(a;q)_n}{(a^{-1}q^{1-n};q)_k}\left(-\frac{q}{a}\right)^kq^{({}^k_2)-nk},
\ee
we have for $a=q^2$ and $q^{2\alpha+2}$,  
\be
\label{ases}
 \tilde{h}_{2n,\alpha}(x,y|q)=\frac{(q;q)_{2n}\, x^{2n}}{(q;q)_{2n,\alpha }}
  \sum_{k=0}^{n}
 \frac{ (-1)^k q^{-4nk+k(2k+1)  } (q^{-2n}, q^{-2n-2\alpha};q^2)_{k} }{ (q^2;q^2)_k   q^{4({}^k_2)-4nk - 2\alpha k}}\left(\frac{y}{x^2}\right)^k.  \nonumber
\ee 
After simplification,   the last equation reads 
\be
  \tilde{h}_{2n,\alpha}(x,y|q)=\frac{(q;q)_{2n}  }{(q;q)_{2n,\alpha}}
x^{2n}\sum_{k=0}^{n}
 \frac{   (q^{-2n}, q^{-2n-2\alpha};q^2)_{k}  }{ (q^2;q^2)_k}\left(-\frac{y \,q^{2\alpha+3}}{x^2}\right)^k.
\ee 
  In the  odd case, the proof follows the same steps as the even case.  
\end{proof}
Now, we are in position to prove    Lemma \ref{thems}.
\begin{proof}{(of Lemma \ref{thems})}  For   $n$ even, the relation (\ref{hyperg1}) becomes:
\be 
\label{asks}
\tilde{h}_{2n,\alpha}(x,y|q)=   
\frac{(q;q)_{2n}}{(q;q)_{2n,\alpha}}\,x^{2n}\,{}_2\Phi_{1}\left(\begin{array}{c}q^{-2n},q^{-2n-2\alpha}\\
0\end{array}\Big|\;q^2;\;-\frac{ y\,q^{2\alpha+3}}{x^2}\right).
\ee
By  taking $a^{-1}=q^{-2\alpha-2}$ and $z=-q^{2n+1}\,x^2y^{-1}$ and the formula \cite[p.17, Eq. (0.6.17)]{ASK}
\be
{}_2\Phi_1\left(\begin{array}{c}q^{-n},a^{-1}q^{1-n}\\
0\end{array}\Big|\;q;\;\frac{a q^{n+1}}{z}\right)=(a;q)_n(qz^{-1})^n {}_1\Phi_1\left(\begin{array}{c}q^{-n}\\
a\end{array}\Big|\;q; z\right)
\ee
we have 
\be
\label{askrs}
  {}_2\Phi_1\left(\begin{array}{c}q^{-2n},q^{-2n-2\alpha}\\
0\end{array}\Big|\;q^2;\;-\frac{ y\,q^{2\alpha+3}}{ x^2}\right)=
\ee
\be
(q^{2\alpha+2};q^2)_n\left(-\frac{y}{x^2}\right)^{n}q^{-2n^2+n}\;{}_1\Phi_1\left(\begin{array}{c}q^{-2n}\\
q^{2+2\alpha}\end{array}\Big|\;q^2;\;-\frac{ q^{2n+1}\,x^2}{ y}\right). \nonumber
\ee 
By using (\ref{Lagu}), the relation (\ref{askrs}) can be written as
\be
\label{results}
 q^{-2n^2+n} \;(q^2;q^2)_n \left(-\frac{y}{x^2}\right)^{n}L_n^{(\alpha)}
\left(x^2y^{-1}q^{-2\alpha-1};q^{2}\right).
\ee
The assertion (\ref{entefunc}) of Lemma \ref{thems} follows  by
summarizing the above calculations in (\ref{asks})-(\ref{results}). \\
In the odd case, the proof follows the same steps as the even case.  
\end{proof}
\begin{proof}{(of Theorem \ref{thfems})} By taking  the limit $\alpha \to +\infty$ in the assertions (\ref{entefunc}) and (\ref{entfunc}) of Lemma \ref{thems}, respectively,  we get the assertions (\ref{asser1}) and (\ref{asser2}) of  Theorem  \ref{thfems}.   
\end{proof}
\section{Connection  formulae for the generalized discrete $q$-Hermite II   polynomials $\{\tilde{h}_{n,\alpha}(x,y|q)\}_{n=0}^\infty$}
We begin this section with the following theorem:
\begin{theorem}
\label{thm1}
The sequence of  gdq-H2P $\{\tilde{h}_{n,\alpha}(x,y|q)\}_{n=0}^\infty$, which is  defined by the relation (\ref{sama:qdoublyH}),  satisfies the connection formula
\be 
\label{thme1}
  \tilde{h}_{n,\alpha}(x, \omega|q)=
 (q;q)_n\sum_{k=0}^{\lfloor\,n/2\,\rfloor}  \frac{q^{-2n k +k(2k+1)  }\,(-\omega\oplus_{q^2} y)^{k } }{(q^2;q^2)_k\,(q;q)_{n-2k}}   \tilde{h}_{n-2k,\alpha}(x,y|q). 
\ee 
\end{theorem}
To prove   Theorem \ref{thm1}, we need the following Lemma.
\begin{lemma}
 The following  generating function   for  gdq-H2P  $\{\tilde{h}_{n,\alpha}(x,y|q)\}_{n=0}^\infty$  holds true.
\bea
\label{gentfunc}
\tilde{e}_{q^2,-\frac{1}{2}}(-yt^2)\tilde{E}_{q ,\alpha}(xt)=\sum_{n=0}^\infty\frac{q^{({}^n_2)}\,t^n}{(q;q)_n}\tilde{h}_{n,\alpha}(x,y|q),\quad |yt|<1.
\eea
\end{lemma}
\begin{proof}
  Let us consider the function
\be 
\label{azdz}
f_q(t;x,y):= \sum_{n=0}^{\infty}\frac{q^{({}^n_2)}\,t^n}{(q;q)_n}\tilde{h}_{n,\alpha}(x,y|q). 
\ee
By replacing   in   (\ref{azdz})   gdq-H2P $\tilde{h}_{n,\alpha}(x,y|q)$ by its explicit expression (\ref{sama:qdoublyH}) we obtain 
\bea 
\label{pasd}
f_q(t;x,y)=\sum_{n=0}^{ 
\infty }t^{n}\left(\sum_{k=0}^{\lfloor\,n/2\,\rfloor}
\frac{ (-1)^k q^{({}^n_2)-2n k+k(2k+1) }  x^{n-2k}\,y^k}{(q;q)_{n-2k,\alpha}\,  (q^2;q^2)_k}\right).
\eea
The right-hand side of  (\ref{pasd})   also reads 
\be 
\sum_{n=0}^{\infty }\sum_{k=0}^{\left\lfloor  {n}/{2}\right\rfloor } 
 \frac{ (-1)^k q^{\binom{n-2k}{2}} (yt^2)^k (xt)^{n-2k}}{(q;q)_{n-2k,\alpha}\,  (q^2;q^2)_k}.
\ee
Next, changing $n-2k$ by $r,\; r=0, 1, \cdots$, the last relation becomes 
\be 
\sum_{n=0}^{\infty }\frac{\left( -yt^2\right) ^{n}}{(q^2;q^2)_{n}}\sum_{r=0}^\infty
\frac{q^{\binom{r}{2}} \left( xt\right) ^{r}}{(q;q)_{r,\alpha}}.    
\ee
Hence,  
\be 
\label{62ee}
f_q(t;x,y)= \tilde{e}_{q^2,-\frac{1}{2}}(-yt^2)\tilde{E}_{q,\alpha}(xt).  
\ee
\end{proof}
Now, we are in position to prove   Theorem \ref{thm1}.
\begin{proof}{(of Theorem \ref{thm1})}
Replacing $t$ by $u \op t$ in (\ref{gentfunc}), we find the following generating function 
\be
\label{samadoublygeneratrice}
\tilde{E}_{q,\alpha}\Big[ (u \op t)x\Big]\tilde{e}_{q^2,-\frac{1}{2}} \Big[-y(u \op t)^2 \Big] =
 \sum_{n=0}^{\infty} \frac{q^{({}^n_2)} (u \op t)^n}{(q;q)_n}  \tilde{h}_{n,\alpha}(x,y|q)
\ee
which by using (\ref{insert}), becomes
\be 
\label{unctionqdoublygeneratrice}
 \tilde{E}_{q,\alpha}\Big[(u \op t)x\Big]=\tilde{E}_{q^2,-\frac{1}{2}} \Big[y(u \op t)^2 \Big]
 \sum_{n=0}^{\infty} \frac{q^{({}^n_2)} (u \op t)^n}{(q;q)_n}  \tilde{h}_{n,\alpha}(x,y|q).
\ee
Replacing $y$ by $\omega$ and (\ref{unctionqdoublygeneratrice}), respectively,   in (\ref{samadoublygeneratrice}),
we get
\be 
\label{oublygeneratrice}
 \sum_{n=0}^{\infty} \frac{q^{({}^n_2)} (u \op t)^n}{(q;q)_n}  \tilde{h}_{n,\alpha} (x,\omega|q) =
\ee
\be
=\tilde{e}_{q^2,-\frac{1}{2}} \Big[-\omega(u \op t)^2 \Big] \tilde{E}_{q^2,-\frac{1}{2}} \Big[y(u \op t)^2 \Big] 
 \sum_{n=0}^{\infty} \frac{q^{({}^n_2)} (u \op t)^n}{(q;q)_n}  \tilde{h}_{n,\alpha} (x,y|q).\nonumber
\ee
By using (\ref{insert}), the last relation reads
\be 
\label{qdoublygeneratrice}
 \sum_{n=0}^{\infty} \frac{q^{({}^n_2)} (u \op t)^n}{(q;q)_n}  \tilde{h}_{n,\alpha}(x,\omega|q) 
\ee
\be
=\tilde{ e}_{q^2,-\frac{1}{2}}\Big[(-\omega\oplus_{q^2} y)(u \op t)^2\Big]
 \sum_{n=0}^{\infty} \frac{q^{({}^n_2)} (u \op t)^n}{(q;q)_n}  \tilde{h}_{n,\alpha}(x,y|q).\nonumber
\ee
According to    (\ref{q-expo}), the right-hand side of (\ref{qdoublygeneratrice}) can be written as  
\be
\label{stees}
 \sum_{r=0}^{\infty} \frac{ (-\omega\oplus_{q^2} y)^{r }(u \op t)^{2r }}{(q^2;q^2)_r} \sum_{n=0}^{\infty} \frac{q^{({}^n_2)} (u \op t)^n}{(q;q)_n}  \tilde{h}_{n,\alpha}(x,y|q).
\ee 
Let us substitute
 $ n +2r =k ~\Longrightarrow~ r  \leq  \lfloor\,k/2\,\rfloor$ in (\ref{stees}), then we have: 
\be
\label{restees}
\sum_{n=0}^{\infty}\left( \sum_{k=0}^{\lfloor\,n/2\,\rfloor} \frac{(q^{\binom{n-2k}{2}}(-\omega\oplus_{q^2} y)^{k } }{(q^2;q^2)_k\,(q;q)_{n-2k}}   \tilde{h}_{n-2k,\alpha}(x,y|q)\right) (u \op t)^{n}.
\ee
Next, replacing (\ref{restees}) in   (\ref{qdoublygeneratrice}), we obtain
\be 
\label{qssdoublygeneratrice}
 \sum_{n=0}^{\infty} \frac{q^{({}^n_2)} (u \op t)^n}{(q;q)_n}  \tilde{h}_{n,\alpha}(x,\omega|q) =
\ee
\be
\sum_{n=0}^{\infty}\left( \sum_{k=0}^{\lfloor\,n/2\,\rfloor} \frac{( q^{\binom{n-2k}{2}}(-\omega\oplus_{q^2} y)^{k } }{(q^2;q^2)_k\,(q;q)_{n-2k}}   \tilde{h}_{n-2k,\alpha}(x,y|q)\right) (u \op t)^{n}.\nonumber
\ee
Finally, on equating the   coefficients of like powers of $ (u \op t)^n/(q;q)_n$ in (\ref{qssdoublygeneratrice}),   we get the desired identity. 
\end{proof}
We have the following special cases of   Theorem \ref{thm1} of particular interest.
\begin{corollary}
 Letting:
\begin{itemize}
\item[(i)] $y=0$ in  the assertion (\ref{thme1}) of   Theorem \ref{thm1},  we get the definition of  gdq-H2P (\ref{sama:qdoublyH}), i.e.,
\be 
\label{thrt1}
  \tilde{h}_{n,\alpha}(x, \omega|q)=
 (q;q)_n\sum_{k=0}^{\lfloor\,n/2\,\rfloor} \frac{(-1)^kq^{-2nk +k(2k+1) }\,x^{n-2k}\,\omega^{k } }{(q^2;q^2)_k\,(q;q)_{n-2k,\alpha}};
\ee 
\item[(ii)] $\omega=0$ in   the assertion (\ref{thme1}) of   Theorem \ref{thm1},  and  using (\ref{def2}), we get the   inversion formula for   gdq-H2P
\be 
\label{inve}
  x^n=
 (q;q)_{n,\alpha}\sum_{k=0}^{\lfloor\,n/2\,\rfloor} \frac{q^{-2nk+3k^2 }\;y^k }{(q^2;q^2)_k\,(q;q)_{n-2k}}   \tilde{h}_{n-2k,\alpha} (x,y|q). 
\ee 
\item[iii)] For $y=1$, the summation formulae (\ref{thme1}) 
can be expressed in terms of generalized discrete $q$-Hermite II polynomials $\tilde{h}_{n,\alpha}(x;q)$. Also,  the summation formulae (\ref{thme1}) 
can be written   in terms of  discrete $q$-Hermite II polynomials $\tilde{h}_n(x;q)$ by choosing  $y=1$ and $\alpha=-1/2.$ 
\end{itemize}
 \end{corollary}

\section{Concluding remarks}
In the previous sections,   we have  introduced     gdq-H2P $\tilde{h}_{n,\alpha}(x,y|q)$  and derived  several properties. Also, we have derived implicit summation formula  for   gdq-H2P $ \tilde{h}_{n,\alpha}(x,y|q) $ by using different analytical means  on their  generating function. This process can be extended  to summation formulae for more generalized forms of $q$-Hermite polynomials.  This study is still in progress.

We note that the generating function of   even   and odd    gdq-H2P  $ \tilde{h}_{n,\alpha}(x,y|q) $  are given by
\be 
\label{last}
 \sum_{n=0}^\infty\frac{(-t^2)^nq^{n(2n-1) }\,}{(q;q)_{2n}}\tilde{h}_{2n,\alpha}(x,y|q)=\frac{q^{\alpha(\alpha+\frac{1}{2})}(q^2;q^2)_\infty}{(q^{2\alpha+2};q^2)_\infty\,}\,\frac{x^{-\alpha}J_{\alpha}^{(2)}(2x q^{-\alpha-\frac{1}{2} };q^2) }{(y\,t^2;q^2)_\infty}\nonumber
\ee
and
\be
\label{lasts}
 \sum_{n=0}^\infty\frac{(-1)^nq^{n(2n+1) } \,t^{2n+1}}{(q;q)_{2n+1}}h_{2n+1,\alpha}(x,y|q)=\frac{q^{\alpha(\alpha+1)}(q^2;q^2)_\infty}{(q^{2\alpha+2};q^2)_\infty\,}\,\frac{x^{-\alpha}J_{\alpha}^{(2)}(2x q^{-\alpha};q^2) }{(y\,t^2;q^2)_\infty}\nonumber
\ee 
where  $J_{\nu}^{(2)}(z;q)$ is the $q$-analogue of the Bessel function  \cite{ASK}.\\
 Indeed, it is well known that from (\ref{gentfunc}),   the generating function of   gdq-H2P $ \tilde{h}_{n,\alpha}(x,y|q) $    is given by 
\be 
\tilde{E}_{q,\alpha}(xt)\tilde{e}_{q^2,-\frac{1}{2}}(-yt^2)=\sum_{n=0}^\infty\frac{q^{n(n-1)/2} t^{n}}{(q;q)_{n}}\tilde{h}_{n,\alpha}(x,y|q)
\ee 
which on separating the power in the right-hand side into their even and odd terms by using the elementary identity
\be
\sum_{n=0}^\infty f(n)=
\sum_{n=0}^\infty f(2n) +\sum_{n=0}^\infty f(2n+1) 
\ee
becomes
\be 
\label{Gegentfunc}
\tilde{E}_{q,\alpha}(xt)\tilde{e}_{q^2,-\frac{1}{2}}(-yt^2)=
\ee
\be
\sum_{n=0}^\infty\frac{q^{n(2n-1)}\, t^{2n}}{(q;q)_{2n}}\tilde{h}_{2n,\alpha}(x,y|q)+\sum_{n=0}^\infty\frac{q^{n(2n+1)}\,  t^{2n+1}}{(q;q)_{2n+1}}\tilde{h}_{2n+1,\alpha}(x,y|q). \nonumber
\ee 
Now replacing $t$ by $i\,t$ in (\ref{Gegentfunc}) and equating the real and imaginary parts of the resultant equation, we get the    generating function of   even   and odd    gdq-H2P $ h_{n,\alpha}(x,y|q)$ as   
\be 
\sum_{n=0}^\infty\frac{(-1)^nq^{n(2n-1)}\,t^{2n}}{(q;q)_{2n}}\tilde{h}_{2n,\alpha}(x,y|q)=Cos_{q,\alpha}(xt)\tilde{e}_{q^2,-\frac{1}{2}}(yt^2)
\ee
and
\be
\sum_{n=0}^\infty\frac{(-1)^nq^{n(2n+1)}\, \,t^{2n+1}}{(q;q)_{2n+1}}\tilde{h}_{2n+1,\alpha}(x,y|q)=Sin_{q,\alpha}(xt)\tilde{e}_{q^2,-\frac{1}{2}}(yt^2) 
\ee 
where the generalized $q$-Cosine and $q$-Sine are defined as:
\bea 
\label{qpo-alpha}
Cos_{q,\alpha}(x):&=&\sum_{k=0}^{\infty} \frac{(-1)^nq^{n(2n-1)}\,x^{2n}}{(q;q)_{2n,\alpha}},
\\
\label{qpo-alphas}
Sin_{q,\alpha}(x):&=& \sum_{k=0}^{\infty} \frac{(-1)^nq^{n(2n+1)}\,x^{2n+1}}{(q;q)_{2n+1,\alpha}}.
\eea
By using (\ref{sam1}) and (\ref{sam2}), respectively, the relations   (\ref{qpo-alpha}) and (\ref{qpo-alphas}) can be   expressed in terms of basic hypergeometric functions     as 
\bea 
\label{cos}
Cos_{q,\alpha}(x)&=&{}_0\Phi_{1}\left(\begin{array}{c}- \\
q^{2\alpha+2}\end{array}\Big|\;q^2;\;-q x^2\right)
\\
\label{coss}
Sin_{q,\alpha}(x)&=&\frac{x}{1-q^{2\alpha+2}}\,\,{}_0\Phi_{1}\left(\begin{array}{c}-\\
q^{2\alpha+4}\end{array}\Big|\;q^2;\;-q^2 x^2\right).
\eea
The $q$-analogue of the Bessel function is defined \cite[p.20, Eq.(0.7.14)]{ASK} by
\be
\label{Lag}
J_{\nu }^{(2)}(z;q)=\frac{(q^{\nu +1};q )_\infty}{(q;q)_\infty}\left(\frac{z}{2}\right)^\nu\,{}_0\Phi_{1}\left(\begin{array}{c}-\\
q^{ \nu+1}\end{array}\Big|\;q;\;- \frac{q^{  \nu+1}z^2}{4}\right)
\ee
from which  the generating functions of   (\ref{cos})  and (\ref{coss}) follow.


\end{document}